\documentclass[a4paper,11pt]{article}
\usepackage[margin=1in]{geometry}
\usepackage[numbers]{natbib}

\usepackage{authblk}
\usepackage{comment}
\usepackage[T1]{fontenc}
\usepackage[noend]{algorithmic}
\usepackage{algorithm}
\usepackage{multirow}
\usepackage{amssymb}

\usepackage{amsmath}

\usepackage{amsthm}
\usepackage{enumerate}
\usepackage{hyperref}
\usepackage[all]{hypcap}
\usepackage[referable]{threeparttablex}
\usepackage{footnote}
\usepackage{pgfplots}
\usepackage{wrapfig}
\usepackage{bigstrut}

\makeatletter
\g@addto@macro{\maketitle}{\@thanks}
\makeatother

\theoremstyle{plain}
\newtheorem{thm}{Theorem}[section]
\newtheorem{cor}[thm]{Corollary}

\newtheorem{lem}[thm]{Lemma}

\newtheorem{Def}[thm]{Definition}

\begin{document}

\title{Near-Optimum Online Ad Allocation for Targeted Advertising}
\date{\vspace{-7ex}}
\author[1]{Joseph (Seffi) Naor\thanks{Supported in part by United States-Israel BSF Grant No. 2010-246 and ISF Grant No. 954/11.}}
\author[2]{David Wajc\thanks{Part of this research conducted while the author was working for Yahoo! Labs, Israel.}}

\affil[1]{CS Dept., Technion \\
	\url{naor@cs.technion.ac.il}}

\affil[2]{CS Dept., Carnegie Mellon University \\
	\url{dwajc@cs.cmu.edu}}

\maketitle

\begin{abstract}
Motivated by Internet targeted advertising, we address several ad allocation problems. Prior work has established these problems admit no randomized online algorithm better than $(1-\frac{1}{e})$-competitive (\citet{karp1990optimal,mehta2007adwords}), yet simple heuristics have been observed to perform much better in practice. We explain this phenomenon by studying a generalization of the bounded-degree inputs considered by \citet{buchbinder2007online}, graphs which we call $(k,d)-bounded$. In such graphs the maximal degree on the online side is at most $d$ and the minimal degree on the offline side is at least $k$. We prove that for such graphs, these problems' natural greedy algorithms attain competitive ratio $1-\frac{d-1}{k+d-1}$, tending to \emph{one} as $d/k$ tends to zero. We prove this bound is tight for these algorithms.

Next, we develop deterministic primal-dual algorithms for the above problems achieving competitive ratio $1-(1-\frac{1}{d})^k>1-\frac{1}{e^{k/d}}$, or \emph{exponentially} better loss as a function of $k/d$, and strictly better than $1-\frac{1}{e}$ whenever $k\geq d$. We complement our lower bounds with matching upper bounds for the vertex-weighted problem. Finally, we use our deterministic algorithms to prove by dual-fitting that simple randomized algorithms achieve the same bounds in expectation. Our algorithms and analysis differ from previous ad allocation algorithms, which largely scale bids based on the spent fraction of their bidder's budget, whereas we scale bids according to the number of times the bidder could have spent as much as her current bid. Our algorithms differ from previous online primal-dual algorithms, as they do not maintain dual feasibility, but only primal-to-dual ratio, and only attain dual feasibility upon termination. We believe our techniques could find applications to other well-behaved online packing problems.
\end{abstract}

\maketitle

\vspace{1mm}
\noindent
{\bf Categories and Subject Descriptors:} F.2.2 Nonnumerical Algorithms and Problems.

\vspace{1mm}
\noindent
{\bf Keywords:} Online Matching, Online Ad Allocation, Targeted Advertising, Sponsored search.

\section{Introduction}

Internet advertising is ubiquitous.~Forecast to surpass the 50 billion dollar/year mark in 2015 in the United States alone, 
it has become, to a large extent, the driving economic force behind much of the content of the world wide web.~How is this advertising space sold and bought?~Most ads fall either under sponsored search or targeted advertising, both of which are sold in what constitute instances of the online ad allocation problem.

In online ad allocation, we are faced with the following problem:~advertisers announce to an advertising platform (e.g.~Yahoo, Google, Microsoft)~what their advertising budgets are, and their bids for an ad to be displayed to every kind of user.~The user ``type'' is determined, for example, by search terms searched, in the case of sponsored search, or user-demographics, in the case of targeted advertising.~When a user visits a web-page with an ad slot managed by the ad platform, the latter needs to decide immediately and irrevocably which (if any) of the advertisers' ads to display to the user.~The advertising platform's goal is to maximize its revenues, despite uncertainty concerning future page-views.~This problem can be formulated as a generalization of online bipartite matching, with advertisers as the offline vertices and ad slots as online vertices.\footnote{In this paper, without loss of generality, we assume advertisers only pay for impressions, and not e.g.~clicks.}
See Section \ref{sec:def} for a formal definition of this and other problems we consider.

The theoretical interest in online allocations can be traced back to 1990, when \citet{karp1990optimal} considered the fundamental problem of bipartite maximum matching in an online setting. In their seminal paper, Karp et al.~proved that randomized online algorithms cannot in general achieve competitive ratio above $1-\frac{1}{e}\approx 0.632$, and presented the \textsc{ranking} algorithm, which matches this upper bound and is thus optimal.\footnote{The original proof of \textsc{ranking}'s competitive ratio was found to contain a mistake nearly twenty years later by Krohn and
Varadarajan, but the algorithm's performance has since been re-proven by \citet{birnbaum2008line}, \citet{goel2008online}, and recently by \citet{devanur2013randomized}.}

The online maximum matching problem was generalized, first by \citet{kalyanasundaram2000optimal}, and later by \citet{aggarwal2011online}, who presented algorithms achieving optimal $1-\frac{1}{e}$ competitive ratio for the $b$-matching and vertex-weighted matching problems, respectively.~The AdWords problem, first proposed by \citet{mehta2007adwords}, is the more general ad allocation problem, but subject to the realistic \emph{small bid assumption}, i.e. assuming every advertiser $i$ has budget $B_i$ much larger than its bids $b_{ij}$. (This assumption is necessary to achieve non-trivial results. See \ref{lem:1-R}).~For this problem too the natural greedy algorithm has competitive ratio $\frac{1}{2}$. Mehta et al.~gave an algorithm for this problem with competitive ratio $1-\frac{1}{e}$. \citet{buchbinder2007online} achieved the same results using an online primal-dual approach.
See \citet{mehta2012online} for an in-depth survey of prior art and techniques used to tackle these problems.

We will address the problems discussed above, but first, we start with motivation.

\subsection{Motivation}\label{subsec:motivation}
As is to be expected of a problem for which a loss of $1/e\approx 36.7\%$ can translate itself to billions of dollars in potential revenue lost yearly, researchers have studied weaker models than the adversarial model for the ad allocation problem, in the hope that these may permit better guarantees.~(See \ref{subsec:related}.) In this paper we revisit the stronger adversarial model, for graphs with structural characteristics met by many ad allocation instances arising from targeted advertising. Specifically, we assume advertisers are interested in a large number of ad slots (at least $k$), and that every ad slot is of interest to a relatively small number of advertisers (at most $d$). As with the small bid assumption $B_i\gg b_{ij}$ for the AdWords problem, assumption of the above structure is not only useful in order to obtain better bounds (as we will show), but also constitutes a reasonable assumption for targeted advertising, for the following twin reasons:
\\
\\
\noindent\textbf{online side:} advertisers typically target their advertising campaigns at \emph{specific} segments of the population (e.g. young Californians who ski often); while these segments may be large in absolute terms, they are mostly small in relative terms (e.g., less than four percent of Californians ski often).~Consequently, users tend to belong to relatively few segments. Coupled with the fact that the number of active campaigns at any given time is limited, this implies a restricted pool of ads that might be displayed to any particular user, justifying the small degree assumption for ad slots.
\\
\\
\noindent\textbf{offline side:} advertisers typically target \emph{large} segments of the population (as in the example above), while not allocating a budget high enough to display ads to all users in a segment. 
Coupled with the fact that every page-view of a particular targeted user corresponds to a vertex in the graph, this implies the high degree assumption on the offline side, and more generally for the ad allocation problem, the assumption that $\sum_{i,j}b_{ij}\geq k\cdot B_i$ for some large $k$.

We call the graphs displaying these characteristics \emph{$(k,d)$-bounded graphs}.

\begin{Def}[$(k,d)$-bounded graphs]
We say a bipartite graph $G=(L,R,E)$ is \emph{$(k,d)$-bounded} if every left vertex $i\in L$ has degree $d(i)\geq k$ and every right vertex $j\in R$ has degree $d(j) \leq d$. For ad allocations, we replace $d(i)\geq k$ with the property $\sum_{j} b_{ij} \geq k\cdot B_i$.
\end{Def}

We concern ourselves with such graphs with $k$ large and $d$ small. For brevity's sake, as all graphs in this paper will be bipartite, we refrain from stating the fact explicitly, and refer to $(k,d)$-bounded graphs as $(k,d)$-graphs henceforth.~As the problems studied in this paper are all maximization problems, we adopt the convention that a lower bound indicates a positive result and an upper bound indicates a negative result.

\subsection{Our Results}\label{results}
By focusing on $(k,d)$-graphs, we justify the observed success of greedy algorithms ``in the wild'', and propose algorithms that are exponentially better, and provably optimal under these structural assumptions. Finally, we leverage our deterministic algorithms to prove simple randomized algorithms achieve the same bounds in expectation. Our results hold for the maximum matching, vertex-weighted matching and AdWords problems (with the exception of the matching upper bound for the latter).~Table \ref{table:Results} delineates our results for these problems on $(k,d)-$graphs. We obtain similar results for the general ad allocation problem, even with large-ish bids (see \ref{thm:R_adwords}.)

\begin{center}
\begin{threeparttable}[ht]
\caption{Best results for general and $(k,d)$-graphs}{
	\label{table:Results}
\centering
    \begin{tabular}{ | c || c | c |}
     \hline
Algorithms & General Graphs & $(k,d)-$Graphs \\  \hline \hline
\multirow{2}{*}{Greedy}\bigstrut & $\frac{1}{2}$ (Tight) & $1-\frac{d-1}{k+d-1}$ (Tight) \\	
& Folklore & This work \\ \hline

\multirow{2}{*}{Deterministic}\bigstrut & $\frac{1}{2}$ (Tight)& $1-(1-\frac{1}{d})^k$ (Tight) \\	
& Folklore & This work \\ \hline

\multirow{2}{*}{Randomized}\bigstrut & $1-\frac{1}{e}$ (Tight)\tnote{$\star$} & $1-(1-\frac{1}{d})^k$ \\	
& \cite{karp1990optimal,birnbaum2008line,goel2008online,aggarwal2011online,mehta2007adwords,buchbinder2007online,devanur2013randomized}
& This work \\ \hline
   \end{tabular}}
   \begin{center}
   \begin{tablenotes}[]
       \centering\item[$\star$]{\footnotesize can be achieved \emph{deterministically} for AdWords.}
   \end{tablenotes}
   \end{center}
\end{threeparttable}
\end{center}

We begin by explaining the empirical success of greedy algorithms for the above problems, proving their loss is proportional to the ratio of the maximal degree in the online side to the minimal degree in the offline side; i.e., their competitive ratio tends to \emph{one} as this ratio tends to zero.~We complement this lower bound with a family of examples for which these algorithms do no better.

\begin{thm}\label{thm:greedy_good}
Greedy algorithms achieve a competitive ratio of $\frac{k}{k+d-1}$ on $(k,d)$-bounded graphs. This analysis is tight for all $k\geq d-1$.
\end{thm}

We improve on the above, designing deterministic algorithms with \emph{exponentially} smaller loss.~We prove this is optimal for deterministic algorithms.

\begin{thm}\label{thm:det_better}
There exist deterministic online algorithms for the unweighted and vertex-weighted matching problems with competitive ratio $1-(1-\frac{1}{d})^k>1-\left(\frac{1}{e}\right)^{k/d}$ on $(k,d)$-bounded graphs. Moreover, these algorithms gain at least a $1-(1-\frac{1}{d})^k$ fraction of the total sum of weights. This is optimal whenever $k\geq d$.
\end{thm}

\begin{cor}\label{cor:struct}(Structural Corollary)
For every bipartite graph $G$ with the minimal degree of its left side at least $\ln c$ times larger than the maximal degree of its right side, $G$ has a matching with at least a $(1-\frac{1}{c})$-fraction of $G$'s left side matched.
\end{cor}

In stating our bounds for general ad allocation, we follow the notation of \citet{buchbinder2007online} and denote the maximum bid-to-budget ratio by $R_{\max}=\max_{(i,j)\in E}\big\{\frac{b_{ij}}{B_i}\big\}$.

\begin{thm}\label{thm:R_adwords}
There exists a deterministic algorithm which gains total revenue at least 
\[
\left(\sum_{i\in L} B_i\right) \cdot \left((1-R_{\max})\cdot\left(1-\left(1-\frac{1}{d}\right)^k\right)\right)
\]
for ad allocation on $(k,d)$-graphs with $k\geq d-1$, and is thus $\big((1-R_{\max})\cdot\big(1-\big(1-\frac{1}{d}\big)^k\big)\big)$-competitive. This is optimal -- no deterministic algorithm can do better for $k\geq d$.
\end{thm}

To contrast our results with the state-of-the-art, we note that the algorithms of \citet{mehta2007adwords,buchbinder2007online,devanur2013randomized} achieve competitive ratio $(1-R_{\max})\cdot \big(1-1/\big(1+R_{\max}\big)^{1/R_{\max}}\big)$. This bound tends to $1-\frac{1}{e}$ from below as $R_{\max}$ tends to zero, but is far from this value for larger $R_{\max}$. Our algorithms fare better whenever $k\geq d$ even for large-ish $R_{\max}$.
As stated in \ref{subsec:motivation}, we expect $k$ to be significantly larger than $d$, but in order to emphasize the strength of our bound, let us assume only that $d/k=R_{\max}$. Table \ref{table:largeish} displays the resulting competitive ratios in this case. Note that in this regime our algorithm is already better at $R_{\max}=\frac{1}{3}$ than prior algorithms are at the limit (i.e. when $R_{\max}\rightarrow 0$).
\vspace{-0.4cm}
\begin{center}
\renewcommand{\arraystretch}{1.2}
\begin{table}[ht]
\caption{Results for Ad Allocation with large-ish bids in $(k,d)$-graphs with $d/k=R_{\max}$}{
\label{table:largeish}
\centering
\begin{tabular}{ |c | c | c | c | c | c | c | c | c | c | c |}
    \hline
$R_{\max}$\bigstrut & $\frac{1}{2}$ & $\frac{1}{3}$ & $\frac{1}{4}$ & $\frac{1}{5}$ & $\frac{1}{6}$ & $\frac{1}{8}$  & $\frac{1}{16}$ & $\frac{1}{32}$ & $\frac{1}{100}$ & $\rightarrow 0$ \\ \hline
State-of-the-art &
0.278 & 0.385 & 0.443 & 0.478 & 0.503 & 0.534 & 0.582 & 0.607 & 0.624 & 0.632 \\ \hline
Our Work&
0.432 & 0.633 & 0.736 & 0.795 & 0.831 & 0.875 & 0.938 & 0.969 & 0.99 & 1 \\ \hline
   \end{tabular}}
\end{table}
\end{center}
\vspace{-0.6cm}

Better still, our algorithms are robust to a few outlying advertisers increasing $R_{\max}$, as the $\left(\sum_i B_i\right)\cdot \left(1-R_{\max}\right)$ term in the above bound is rather $\left(\sum_i B_i - \max_{j\in N(i)} b_{ij}\right)$. This is the first such result in the adversarial setting. To the best of our knowledge only the algorithm of \citet{devanur2012asymptotically} for the iid model holds this desired property.
Likewise, our algorithms are robust to few outlying advertisers making the input not $(k,d)$-bounded (alternatively, increasing $k$), as the following theorem asserts.
\begin{thm}[Outliers]\label{thm:outliers}
	If every advertiser $i$ satisfies $\sum_{j} b_{ij}\geq k\cdot B_i$, except for a subset $S\subset L$ with total budget at most an $\alpha$-fraction of the total sum of budgets, $\sum_{i\in S} B_i \leq \alpha \cdot \sum_{i\in L} B_i$, then the algorithms of Theorems \ref{thm:det_better} and \ref{thm:R_adwords} gain revenue at least $(1-\alpha)$ times the bounds guaranteed by the above theorems. In particular, these algorithms achieve competitive ratio at least $(1-\alpha)\cdot\big(1-\big(1-\frac{1}{d}\big)^k\big)$ and $(1-\alpha)\cdot \big((1-R_{\max})\cdot\big(1-\big(1-\frac{1}{d}\big)^k\big)\big)$.
\end{thm}

Finally, we prove that several easy-to-implement randomized algorithms match the bounds of our optimal deterministic algorithms in expectation, despite making no use of the input's structure.

\begin{thm}\label{thm:rand_best?}
Several simple randomized algorithms achieve expected competitive ratio matching those of Theorems \ref{thm:det_better}, \ref{thm:R_adwords}, and \ref{thm:outliers}.
\end{thm}

\subsection{Techniques}\label{subsec:techniques}
As many previous ad allocation algorithms, our algorithms can be seen as bid-scaling algorithms. That is, matches are chosen greedily based on the bids $b_{ij}$ of each advertiser $i$, times a scaling factor. However, contrary to previous algorithms \citet{mehta2007adwords,buchbinder2007online,devanur2013randomized} that scale bids according to the function $1-e^{f(i)-1}$, where $f(i)$ is the fraction of $i$'s budget spent so far, our algorithms essentially scale bids according to an exponential in $u$, the number of \emph{unused opportunities} for spending the current bid $b_{ij}$, specifically, $\big(\frac{d}{d-1}\big)^{u}$. Other differences can be seen in our algorithms' primal-dual interpretation: we make no use of the ad slots' dual variables, and only update the dual variables of each arriving ad slot's neighbors. Interestingly, our online primal-dual algorithms do not guarantee dual feasibility throughout their execution, but only upon termination. To the best of our knowledge, ours are the first online primal-dual algorithms with this behavior.

The above approach works directly for vertex-weighted matching. To generalize our approach to ad allocations, we first consider an intermediary problem -- equal-bids ad allocation -- where every advertiser $i$ bids the same bid $b_i$ for all neighbors $j\in N(i)$.~We reduce this problem in $(k,d)$-graphs to the vertex-weighted problem in $(k,d)$-graphs in an online manner. We then rewrite this reduction along with our vertex-weighted online algorithm as a single online primal-dual algorithm for the equal-bids problem. Guided by this algorithm we devise a primal-dual algorithm for general-bids ad allocation on $(k,d)$-graphs, using a bounded fraction of the advertisers' dual variables to guide our choice of matches and dual updates. This allows us to simulate the bid-scaling described above also in the case where each advertiser has different bids.

Finally, our randomized results stem from our deterministic primal-dual algorithms, whose dual updates we use in our dual-fitting analysis of the randomized algorithms. Dual feasibility follows as it does for our algorithms.~The dual costs are bounded in expectation by the primal cost times the required constant, conditioned over the random algorithm's previous choices.~Taking total expectation over the possible previous choices yields the expected competitive ratio.

\subsection{Related Work}\label{subsec:related}
Several stochastic models have been studied for the problems we address.~Most prominent among these are the random arrival order and i.i.d model with known/unknown distribution. Our algorithms beat all of these bounds \emph{in the worst case} for sufficiently small $d/k$ and $R_{\max}$, replacing stochastic assumptions by structural ones.

For the random order model a line of work beginning with \citet{goel2008online} has shown the optimal competitive ratio for maximum matching lies in the range $(0.696, 0.823)$ \citet{feldman2009online,karande2011online,mahdian2011online,manshadi2012online}.
For the known distribution model \citet{feldman2009online} were the first to show the optimal competitive ratio is strictly greater than $1-\frac{1}{e}$ and bounded away from 1. Subsequent work \citet{bahmani2010improved,haeupler2011online,jaillet2013online} showed the optimal competitive ratio for bipartite matching in this setting lies in the range $(0.706, 0.823)$, and $(0.729,0.823)$ if the expected number of arrivals of each ad slot type is integral. For the vertex-weighted problem under the previously-mentioned integrality assumptions \citet{haeupler2011online} and \citet{jaillet2013online} showed a lower bound of $0.667$ and $0.725$, respectively.
For the AdWords problem under the random order model, \citet{devanur2009adwords} give a $(1-\epsilon)$-competitive algorithm, assuming the online side's size is known in advance and no bid is higher than roughly  $\epsilon^3/|L|^2$ times the optimum value.
\citet{devanur2012asymptotically} gave an algorithm in the unknown distribution model achieving asymptotically optimal competitive ratio of $1-O(\sqrt{R_{\max}})$.

In a different vein, \citet{mahdian2007allocating} considered the AdWords problem given black-box estimates of the input.~They show how to obtain performance trading-off between the worst-case optimal and the black-box's performance on the given input. We require no such algorithm be available, but rather rely on domain-specific structure.

Closer to our work, \citet{buchbinder2007online} considered $(1,d)$-graphs for equal-bids ad allocation. 
We obtain more general results, and strictly better bounds for all $k>d$.

\newpage
\subsection{Paper Outline}
In Section \ref{sec:def} we formally define the problems considered throughout the paper. In Section \ref{sec:greedy} we give a tight analysis of algorithm \textsc{greedy} in $(k,d)$-graphs. In Section \ref{sec:high_deg} we build on the hard examples of Section \ref{sec:greedy} and present optimal algorithms for the online maximum matching and vertex-weighted matching problems in $(k,d)$-graphs. In Section \ref{sec:ad-allocation} we extend these results to the general ad allocation problem. In Section \ref{sec:upper-bounds} we present hardness results for the problems considered. In Section \ref{sec:randomization} we extend our analysis to prove competitiveness of several simple randomized algorithms. We conclude with a discussion of future work and open questions in Section \ref{sec:discuss}.
\section{Problem Definitions}\label{sec:def}
An instance of the \emph{ad allocation} problem consists of a bipartite graph $G=(L, R, E)$. The left-hand $L$ side corresponds to advertisers, and the right-hand side $R$ to ad slots. Each advertiser $i\in L$ has some budget $B_i$ and is willing to bid some value $b_{ij}\leq B_i$ for every neighboring ad slot $j\in N(i)$ (the bids of advertiser $i$ need not be equal for all $j\in N(i)$). Each ad slot $j\in R$ can be allocated to (up to) one advertiser $i$, yielding a profit of $b_{ij}$. The bids for ad slots allocated to an advertiser $i$ may not exceed $i$'s budget, $B_i$. Figure \ref{fig:primaldual-auction} presents the ad allocation problem's LP relaxation and its dual.

\begin{figure}[h]
	\label{fig:primaldual-auction}
	\begin{small}
		\begin{center}
			\begin{tabular}{rl|rl}
				\multicolumn{2}{c|}{Primal (Packing)} & \multicolumn{2}{c}{Dual
					(Covering)} \\ \hline  
				maximize & $\sum_{(i,j)\in E} b_{ij}\cdot x_{ij}$ &
				minimize & $\sum_{i\in L}{B_i \cdot z_i} +  \sum_{j\in R}{y_j}$ \\
				subject to: & & subject to: & \\
				
				$\forall j\in R$: &  $\sum_{(i,j)\in E}x_{ij} \leq 1$ &
				$\forall (i,j)\in E$: & $b_{ij}\cdot z_i + y_j \geq b_{ij}$ \\
				
				$\forall i \in L$: & $\sum_{(i,j)\in E} b_{ij}\cdot x_{ij} \leq
				B_i$ & $\forall i\in L$: & $z_i\geq 0$ \\
				
				$\forall (i,j)\in E$: & $x_{ij} \geq 0$ & $\forall j\in R$: & $y_j \geq0$
			\end{tabular}
		\end{center}
	\end{small}
	\caption{The fractional ad allocation LP and the corresponding dual}
\end{figure} 

An instance of the \emph{online} ad allocation problem consists of an ad allocation instance; the advertisers given up-front, along with their budgets, and the ad slots arriving one-by-one, together with their edges and bids. An online ad allocation algorithm must, upon arrival of an ad slot $j$, determine to which advertiser (if any) to allocate the ad slot. Allocations are irrevocable, and so must be made to \emph{feasible} advertisers, whose residual budget is sufficient to pay their actual bid.

We will consider several interesting special cases of the above problem throughout this paper. These problems are both interesting in their own right (theoretically as well as practically), in addition to providing some insight towards achieving a solution to the general problem.

The \emph{equal-bids} online ad allocation problem is the above problem with each advertiser $i$ bidding the same value for all neighboring ad slots; i.e., $b_{ij}=b_i$ for all $j\in N(i)$.

The online \emph{vertex-weighted matching} problem is the above problem with every advertiser $i$ bidding all its budget for every neighboring ad slot; i.e., $b_{ij}=B_i$ for all $j\in N(i)$.

The online \emph{maximum matching} problem is the above problem with all budgets and bids equal to 1; i.e., $b_{ij}=B_i=1$ for all $j\in N(i)$.

\section{Warm-up: Greediness in $(k,d)$-Graphs}\label{sec:greedy}

In this section we show that the natural greedy algorithms for the problems considered, which in general graphs are only 1/2-competitive, achieve on $(k,d)$-graphs a competitive ratio tending to one as $d/k$ tends to zero. We prove this result by applying dual-fitting, and prove our analysis is tight. 

Algorithm \textsc{greedy} for the online ad allocation problem matches an ad slot $j\in R$ to a feasible neighbor $i$ with highest bid $b_{ij}$. For the vertex-weighted case, where $b_{ij}=B_i$ for all $j\in N(i)$, this reduces to picking an unmatched neighbor of highest weight. Our analysis relies on a dual-fitting formulation, given in Algorithm \ref{alg:adwords-greedy} below.

\begin{algorithm}
	\caption{\textsc{ad allocation greedy} (Dual-Fitting Formulation)}
	\label{alg:adwords-greedy}
	\begin{algorithmic}[1]
		\medskip
		\STATE \textbf{Init:} set $z_i\leftarrow 0\  \text{ for all } i\in L$
		\FORALL{$j\in R$}
			\IF{$j$ has a feasible neighbor}
				\STATE match $j$ to a feasible neighbor maximizing ${b_{ij}}$
				\STATE set $x_{ij}\leftarrow 1$
				\STATE set $z_i\leftarrow \min\{1,z_{i}+\frac{b_{ij}}{B_i}\}$\label{line:match-inc}
				\STATE set $z_{i'}\leftarrow \min\{1,z_{i'}+\frac{b_{i'j}}{k\cdot B_i}\}$ for every feasible neighbor of $j$, $i'\neq i$
		\ENDIF
		\ENDFOR
		\FORALL{$i\in L$}\label{line:prob}
			\IF{$i$'s residual budget is less than $R_{\max}\cdot B_i$}
			\STATE set $z_i\leftarrow 1$\label{line:prob_end}
		\ENDIF
		\ENDFOR
	\end{algorithmic}
\end{algorithm}

\begin{thm}\label{thm:match-greedy}
Algorithm \textsc{greedy} is  $\left(\frac{k}{k+d-1}\right)$-competitive for the unweighted, vertex-weighted maximum matching and equal-bids ad allocation problems on $(k,d)$-graphs.
\end{thm}
\begin{thm}\label{thm:add-allocation-greed}
	Algorithm \textsc{greedy} is  $\frac{(1-R_{\max})\cdot k}{k+(d-1)\cdot(1-R_{\max})}>(1-R_{\max})\cdot\frac{k}{k+d-1}$ competitive for online ad allocation on $(k,d)$-graphs with $k\geq 1$ and $R_{\max}=\max_{(i,j)\in E} \{b_{ij}/B_i\}<1$.
\end{thm}
\begin{proof}
We prove the following claims: (a) $z,y$ form a feasible dual solution (b) for every $j\in R$ the changes to the primal and dual solutions' values, $\Delta P$ and $\Delta D$, satisfy $\Delta D/\Delta P\leq \frac{k+d-1}{k}$. (c) for the vertex-weighted and unweighed matching problems and equal-bids problem Lines \ref{line:prob}-\ref{line:prob_end} incur no dual cost, and (c') for the general ad allocation problem Lines \ref{line:prob}-\ref{line:prob_end} cost the dual solution no more than $R_{\max}/(1-R_{\max})$ times the primal profit. As $x$ forms an integral feasible primal solution, claims (a,b,c) combined entail Theorem \ref{thm:match-greedy}. Similarly, claims (a,b,c') entail Theorem \ref{thm:add-allocation-greed}, as claims (b) and (c') imply the ratio of the programs' overall values is at least 
\[
\frac{P}{D}\geq \frac{P}{\frac{k+d-1}{k}\cdot P + \frac{R_{\max}}{1-R_{\max}}\cdot P}=\frac{(1-R_{\max})\cdot k}{k+(d-1)\cdot(1-R_{\max})}
\]

\textbf{Claim (a):} For every advertiser $i\in L$, if over a $(1-R_{\max})$-fraction of $i$'s budget is spent then $z_i$ is set to one in Line \ref{line:prob_end}. Otherwise, $i$ is a feasible match of all of its neighbors $j$, each such $j$ causing $z_i$ to increase by at least $\frac{b_{ij}}{k\cdot B_i}$. As $\sum_j b_{ij}\geq k\cdot B_i$ then $z_i=1$ by the algorithm's termination. Consequently, all dual inequalities are satisfied. 

\textbf{Claim (b)}: For each ad slot $j\in R$, by the choice of $j$'s match $i$, and the fact that $j$ has degree $d(j)$ at most $d$, the primal value increases by $\Delta P=b_{ij}$ and the dual cost increases by at most $\Delta D=b_{ij} + \sum_{i'\in F_j\setminus\{i\}} b_{i'j}/k\leq b_{ij}\cdot(1+\frac{d-1}{k})$.

\textbf{Claim (c):} For an advertiser $i\in L$ to have spent over $(1-R_{\max})B_i$ for all but the general problem, it must and have $z_i$ set to one. Thus Lines \ref{line:prob}-\ref{line:prob_end} incur no dual cost.

\textbf{Claim (c'):} For an advertiser $i\in L$ to be affected by Lines \ref{line:prob}--\ref{line:prob_end}, it must spend up to a $(1-R_{\max})$-fraction of its budget. However, whenever $i$ spends an $f$-fraction of its budget, the dual variable $z_i$ increases by $f$ in Line \ref{line:match-inc}, and so the cost of increasing $z_i$ in line \ref{line:prob_end} is at most $R_{\max}\cdot B_i$, while $i$ garnered a primal profit of at least $(1-R_{\max})\cdot B_i$. The total dual cost of Lines \ref{line:prob}--\ref{line:prob_end} is thus at most $\frac{R_{\max}}{1-R_{\max}}\cdot P$, for $P$ the primal profit.
\end{proof}

\subsection{Tight Examples for Algorithm \textsc{greedy}}
\label{sec:greedy_tight}
We show that our analysis of algorithm \textsc{greedy} for the unweighted and vertex-weighted matching is tight whenever $k\geq d-1$ .\footnote{For $k<d-1$ the $\frac{k}{k+d-1}$ bound is strictly less than the $\frac{1}{2}$ bound obtained by all maximal matchings, and so the bound cannot be tight for $k<d-1$. We therefore turn our attention to the case $k\geq d-1$.}
\begin{thm}
For all $k\geq d-1$ there exist $(k,d)-$graphs $G$ with maximal matchings that achieve a competitive ratio no better than $\frac{k}{k+d-1}$ on $G$.
\end{thm}

\begin{table}[h]
\begin{wrapfigure}{r}{0.3\textwidth}
	\begin{center}
		\capstart
		\vspace{-0.5cm}
		\includegraphics[scale=0.45]{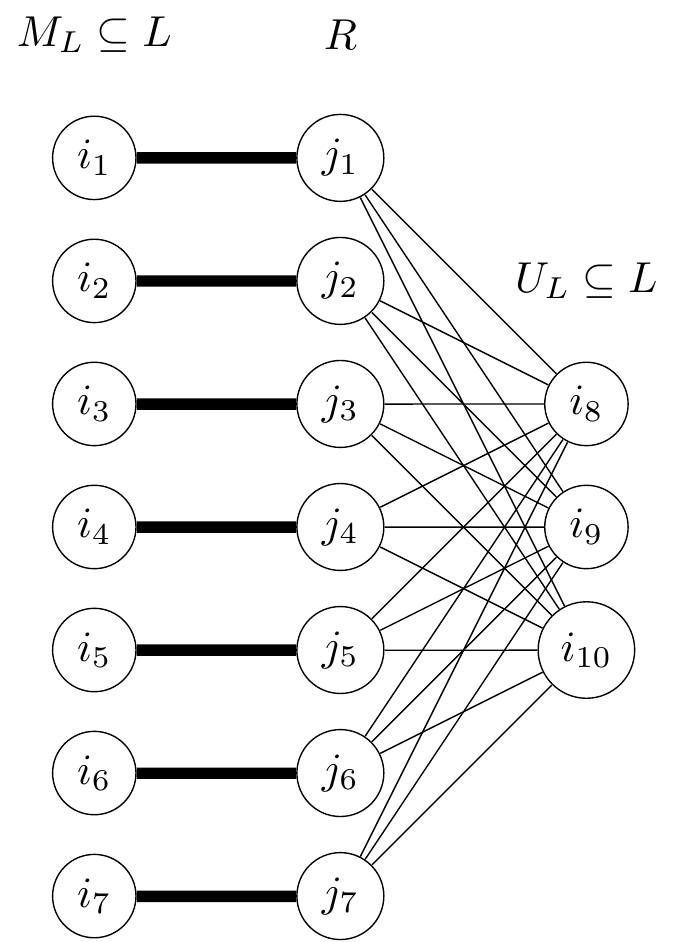}
		\caption{The graph and matching after the arrival of the first $k$ ad slots, for $k=7$ and $d=4$. Matching edges are marked in \textbf{bold}.}
		\label{fig:tight_maximal}
	\end{center}
\end{wrapfigure}

\vspace{-0.4cm}

\emph{Proof.} The tight example, along with a poor choice of matching, is defined as follows: The offline side contains $k+d-1$ advertisers. The first $k$ ad slots by order of arrival $j_1,j_2,\dots,j_k$ each have degree exactly $d$, with the $t$-th ad slot $j_t$ having as neighbors the $t$-th advertiser $i_t$ (to which it is matched), as well as the last $d-1$ advertisers,  $i_{k+1},i_{k+2},\dots,i_{k+d-1}$. After the first $k$ ad slots arrive all $d-1$ unmatched advertisers have degree exactly $k$. (By this stage the ad slots and the last $d-1$ advertisers form a copy of $K_{k,d-1}$. See Figure \ref{fig:tight_maximal}.) The following ad slots are used to increase the degree of the first $k$ advertisers to at least $k$, while guaranteeing that the first $k$ advertisers can be simultaneously matched to the last ad slots (for example, by having all the latter ad slots have degree one).~The resulting graph is a $(k,d)-$graph with $k$ of its advertisers matched for which all $k+d-1$ advertisers can be matched simultaneously. \qed
\end{table}

For any $R_{\max}$ a unit fraction, gluing $1/R_{\max}$ copies of the above tight example at the advertisers, with each advertiser having a budget of $1/R_{\max}$, yields an equal-bid ad allocation instance and greedy allocation for which the same $\frac{k}{k+d-1}$ performance holds, proving tightness of our analysis for equal-bid allocations. We now state a theorem implying our analysis' tightness for \textsc{greedy} in general ad allocations. For a proof of this theorem, see Appendix \ref{sec:greedy-hard}.

\begin{thm}
	For all $k\geq d-1$ and $R_{\max}\leq \frac{1}{2}$ there exist $(k,d)-$graphs $G$ and greedy allocations achieving a competitive ratio at most $\frac{(1-R_{\max})\cdot k}{k+(d-1)\cdot(1-R_{\max})}$ on $G$.
\end{thm}
\vspace{-0.4cm}
\section{Optimal Vertex-Weighted Matching on $(k,d)$-graphs}
\label{sec:high_deg}
The previous section shows our analysis of \textsc{greedy} is tight, though for a particular(ly bad) instantiation of the input, and more importantly of the algorithm.~The family of tight examples suggests the following improved algorithm: match every arriving ad slot to an unmatched neighbor of highest degree.~This algorithm, which we call \textsc{high-degree}, is given below.~The intuition behind it, substantiated by the above examples, is that unmatched advertisers with higher degree may have fewer chances to be matched later.~This approach fares better on the above examples (actually yielding an \emph{optimal} solution), but can it do better than \textsc{greedy} for \emph{all} $(k,d)-$graphs?~We answer this question in the affirmative, proving a lower bound with \emph{exponentially} smaller loss. In Section \ref{sec:upper-bounds} we prove a matching lower bound, implying the algorithm's \emph{optimality}.

\begin{algorithm}[h]
\caption{\textsc{high-degree}}
\label{alg:high_deg}
\begin{algorithmic}[1]
\medskip
\FORALL{$j\in R$}
	\IF{$j$ has an unmatched neighbor}
		\STATE match $j$ to unmatched neighbor of highest degree.
	\ENDIF
\ENDFOR
\end{algorithmic}
\end{algorithm}

\subsection{Analysis of \textsc{high-degree}}
\begin{thm}\label{thm:high_lb}
Algorithm \textsc{high-degree} is
$1-(1-\frac{1}{d})^k$ competitive for all $(k,d)-$graphs.
\end{thm}

\begin{cor}\label{cor:regular}
On $d$-regular graphs algorithm \textsc{high-degree} is $1-(1-\frac{1}{d})^d$ competitive.
\end{cor}
This is the first result for maximum online matching in regular graphs in the adversarial setting, beating the $1-\frac{1}{e}$ ``barrier" \emph{deterministically}.

Theorem \ref{thm:high_lb} can be proven directly (see \ref{sec:potential}), but in order to set the groundwork for proofs of our more general results, we generalize this algorithm and rewrite it as a primal-dual algorithm, below (\ref{alg:vw-high-pd}). The constant $C$ will be chosen during the analysis.

\begin{algorithm}
	\caption{Vertex-Weighted \textsc{high-degree} (Primal-Dual Formulation)}
	\label{alg:vw-high-pd}
	\begin{algorithmic}[1]
		\medskip
		\STATE \textbf{Init:} set $z_i\leftarrow 0\  \text{ for all } i\in L$.
		\FORALL{$j\in R$}
		\IF{$j$ has an unmatched neighbor $i$}
		\STATE match $j$ to an unmatched neighbor $i$ maximizing $(z_i+C)\cdot b_{ij}$.
		\STATE set $x_{ij}\leftarrow 1$.
		\STATE set $z_i\leftarrow 1$.
		\STATE set $z_{i'}\leftarrow \min\{1,z_{i'}\cdot\big(\frac{d}{d-1}\big) +\frac{1}{d-1}\cdot C\}$ for every feasible neighbor of $j$, $i'\neq i$.
		\ENDIF
		\ENDFOR
	\end{algorithmic}
\end{algorithm}

\begin{thm}\label{thm:vw-high-deg}
Algorithm \ref{alg:vw-high-pd} generalizes \textsc{high-degree} and is $1-\big(1-\frac{1}{d}\big)^k$ competitive. Moreover, it gains revenue at least $(1-\big(1-\frac{1}{d}\big)^k)\cdot \big(\sum_i B_i\big)$.
\end{thm}

\begin{proof} 
We rely on the following observation, verifiable by induction:
All unmatched advertisers $i$ satisfy $z_i= C\cdot(\big(\frac{d}{d-1}\big)^{d(i)}-1)$.
Hence Algorithm \ref{alg:vw-high-pd} matches each ad slot $j$ to an unmatched neighbor $i$ maximizing $b_{ij}\cdot C\cdot(\big(\frac{d}{d-1}\big)^{d(i)}-1)$. For the unweighted problem, $b_{ij}=1$. By monotonicity of exponentiation, picking such $i$ is tantamount to picking an advertiser of highest degree. We proceed to bound the algorithm's gain.

Let $j\in R$ be some ad slot matched to $i$. The incurred change to the primal profit equals $\Delta P=b_{ij}$. By our choice of $j$'s match, the change to the dual cost satisfies
\[
\begin{array}{llr}
\Delta D & = (1-z_i)\cdot b_{ij}+\sum_{i'\in N(j)\setminus\{i\}}\left(\left(\frac{1}{d-1}\right)\cdot (z_{i'}+C)\cdot b_{i'j}\right) \\

& \leq (1-z_i)\cdot b_{ij} + (d-1)\cdot\left(\frac{1}{d-1}\right)\cdot(z_i+C)\cdot b_{ij}\\

& = (1+C)\cdot b_{ij}.\\
\end{array}
\]

Given dual feasibility, the above would imply a competitive ratio of $1/(1+C)$. Hence, we choose the minimal $C$ ensuring $z_i=1$ by the algorithm's end for all advertisers $i$ (matched and unmatched alike). Recall all unmatched advertisers $i$ satisfy $z_i= C\cdot((\frac{d}{d-1})^{d(i)}-1)$. 
As such $i$ have degree at least $k$ by the algorithm's end (but possibly no higher), the minimal $C$ ensuring $z_i=1$ is $C=1/((\frac{d}{d-1})^k-1)$. As the dual solution has $z_i=1$ for all $i$ by the algorithm's termination, the dual cost is exactly $D=\sum_{i\in L}B_i$. Consequently, the primal gain satisfies $P\geq \frac{1}{1+C}\cdot \big(\sum_i B_i\big)$. The theorem follows.
\end{proof}

The above algorithm implies structural Corollary \ref{cor:struct} and the following corollary.

\begin{cor}\label{cor}
	For $(k,d)$-graphs with $k\geq d\cdot \ln {|L|}$, by integrality of number of vertices matched, \textsc{high-degree} successfully matches \emph{all} of $L$, obtaining a \emph{maximum} matching.
\end{cor}

We can extend our analysis to handle the possible existence of \emph{outlying} advertisers $i$, that do not satisfy $\sum_{j} b_{ij}\geq k\cdot B_i$, and so may not satisfy $z_i=1$, ruining dual feasibility. Let $S\subseteq L$ be the set of outlying advertisers, and assume $\sum_{i\in S} B_i \leq \alpha \cdot \sum_{i\in L} B_i$. As $z_i=1$ for all $i\not\in S$, we have $D\geq (1-\alpha)\cdot \sum_{i\in L} B_i$, implying the following theorem.

\begin{thm}[Outliers]\label{thm:outliers_vw}
	Let $S\subseteq L$ be the set of outlying advertisers, and $\alpha$ be such that $\sum_{i\in S} B_i \leq \alpha \cdot \sum_{i\in L}$. Then Algorithm \ref{alg:vw-high-pd} gains at least $(1-\alpha)\cdot (1-\big(1-\frac{1}{d}\big)^k)\cdot (\sum_i B_i)$, and in particular is $(1-\alpha)\cdot (1-\big(1-\frac{1}{d}\big)^k)$-competitive.	
\end{thm}

\subsection{Potential-based Analysis of \textsc{high-degree}}\label{sec:potential}
In this subsection we present a potential-based proof of Theorem \ref{thm:high_lb}. We note that this proof can easily be extended to provide alternative proofs of Theorems \ref{thm:vw-high-deg} and \ref{thm:outliers_vw}.

\begin{thm}\label{thm:strong-high-deg}
Algorithm \textsc{high-degree} achieves value at least
$\big(1-(1-\frac{1}{d})^k\big)\cdot |L|$ for all $(k,d)-$graphs $G=(L,R,E)$, and it is therefore $\big(1-(1-\frac{1}{d})^k\big)$-competitive.
\end{thm}

\begin{proof}
Let $U_L\subseteq L$ denote the set of unmatched advertisers. Consider the following potential:
\begin{center}
$\phi = \underset{i\in U_L}{\sum} \left(\frac{d}{d-1}\right)^{d(i)}.$
\end{center}
Algorithm \textsc{high-degree} outputs a matching that effectively strives to greedily minimize $\phi$.\footnote{The sum of unmatched advertisers' degrees may seem like a more natural potential function to consider, but it turns out that it cannot be used to derive tight bounds. E.g.,~it does not yield a bound significantly better than $\frac{5}{8}=0.625$ for $k=d$.} The initial and final values of the potential function hold $\phi_{start}=|L|$ and $\phi_{final}\geq \left(d/(d-1)\right)^k\cdot |U_L|$, respectively. Denote by $\Delta\phi_j$ the change to $\phi$ incurred by the arrival of ad slot $j\in R$. Clearly, if $j$ is unmatched we have $\Delta\phi_j=0$. On the other hand, if $j$ is matched to a neighbor $i$, previously of degree $d(i)$, we find that $i$'s matching results in $\phi$ decreasing by $\left(d/(d-1)\right)^{d(i)}$, and the degree of $j$'s remaining unmatched neighbors increase each cause $\phi$ to increase by at most $\left(d/(d-1)\right)^{d(i)+1}-\left(d/(d-1)\right)^{d(i)}$. Therefore, if $j$ is matched to $i$ we have
\[
	\begin{array}{llr}
		\Delta\phi_j & \leq - \left(\frac{d}{d-1}\right)^{d(i)} + (d-1)\cdot\left(\left(\frac{d}{d-1}\right)^{d(i)}\cdot(\frac{d}{d-1}-1)\right) \\
		
		& = -\left(\frac{d}{d-1}\right)^{d(i)} + \left(\frac{d}{d-1}\right)^{d(i)} = 0.\\
	\end{array}
\]

In other words $\Delta\phi_j\leq 0$, irrespective of whether or not $j$ is matched. By this fact and our bounds on the initial and final potential, we  find that
\begin{equation*}
\left(\frac{d}{d-1}\right)^k\cdot |U_L|\leq \phi_{final} \leq \phi_{start} = |L|.
\end{equation*}
The theorem follows.
\end{proof}

\section{Online Ad Allocation}\label{sec:ad-allocation}
In this section we solve the ad allocation problem. We consider first the equal-bids case, where each advertiser $i$ offers the same bid for all its neighbors; i.e., $b_{ij}=b_i$ $\forall j\in N(i)$. This will prove to be a useful stepping-stone towards a solution for general bids, in \ref{subsec:gen-bids}.

One way to solve equal-bids ad allocation is via an online reduction to vertex-weighted matching in $(k,d)$-graphs. As each advertiser $i$ bids $\sum_{j\in N(i)}b_i \geq k\cdot B_i$ in total, we have $d(i)\geq k\cdot B_i/b_i$. Without loss of generality, $B_i/b_i$ is integral. 
The reduction splits each $i$ into $B_i/b_i$ copies, each of value $b_i$ and receiving up to $k$ distinct edges of $i$, stopping if the copy is matched. The obtained graph $G$ is $(k,d)$-bounded (perhaps after adding inconsequential neighbors to matched advertisers), and matchings in $G$ induce allocations of same value for the ad allocation instance. As Algorithm \ref{alg:vw-high-pd} gains $1-(1-\frac{1}{d})^k$ of the sum of vertex weights, or equivalently the sum of budgets, applying it yields a $1-(1-\frac{1}{d})^k$ competitive solution to the original ad allocation instance.

We restate the above as a primal-dual algorithm for equal-bids ad allocation. (See Algorithm \ref{alg:primal-dual-ad-allocation} below). In this algorithm, $z_i^c$ serves the role of $z_i$ in Algorithm \ref{alg:vw-high-pd} for $i$'s ``current copy'' (hence the $c$ in the notation), weighted to reflect the copy contributes $b_i/B_i$ of $i$'s budget. Intuitively, when $i$ is matched we imagine its current copy is matched, and set $z_i^c$ to $b_i/B_i$. Conversely, we ensure that once the copy has $k$ edges $z_i^c=b_i/B_i$. Either way, once $z_i^c=b_i/B_i$, we add $z_i^c$ to $z_i$ and nullify $z_i^c$ (moving to $i$'s next copy, whose dual variable would be zero in Algorithm \ref{alg:vw-high-pd}.) The number of copies of $i$ guarantees dual feasibility and the choice of match and dual updates guarantee the desired bound.

\begin{algorithm}
	\caption{Equal-Bid Ad Allocation in $(k,d)$-graphs}
	\label{alg:primal-dual-ad-allocation}
	\begin{algorithmic}[1]
		\medskip
		\STATE \textbf{Init:} set $z_i\leftarrow 0\ , z_i^c\leftarrow 0\ \text{ for all } i\in L$
		\FORALL{$j\in R$}
		\IF{$j$ has a feasible neighbor $i$}
			\STATE match $j$ to feasible neighbor $i$ maximizing $z_i^c\cdot B_i+C\cdot b_i$.
			\STATE set $x_{i,j}\leftarrow 1$.
			\STATE set $z_i^c\leftarrow b_i/B_i$.\label{line:bar1}
			\FORALL{feasible neighbor of $j$, $i'\neq i$}
				\STATE set $z_{i'}^c\leftarrow \min\{ b_{i'}/B_{i'},\ z_{i'}^c\cdot\big(\frac{d}{d-1}\big) +\frac{1}{d-1}\cdot C\cdot b_{i'}/B_{i'}\}$ \label{line:bar2}
			\ENDFOR
			\FORALL{$i'\in N(j)$ with $z_{i'}^c=b_{i'}/B_{i'}$}
				\STATE set $z_{i'}\leftarrow z_{i'} + z_{i'}^c$.
				\STATE set $z_{i'}^c\leftarrow 0$.
			\ENDFOR
		\ENDIF
		\ENDFOR
	\end{algorithmic}
\end{algorithm}

\vspace{-0.4cm}
\begin{thm}\label{thm:equal-bids}
	Algorithm \ref{alg:primal-dual-ad-allocation} with $C=1/\big(\big(\frac{d}{d-1}\big)^k-1\big)$ gains revenue $\left(1-(1-\frac{1}{d})^k\right)\cdot \sum_i B_i$, and is thus $\left(1-(1-\frac{1}{d})^k\right)$-competitive for the equal-bid problem on $(k,d)$-graphs.
\end{thm}

\begin{proof} To bound the primal-dual ratio, we bound increases of $z_i^c\cdot B_i$, as all dual costs can be traced back to past increases of $z_i^c$. Consider some ad slot $j$ matched to $i$. The primal gain is $\Delta P=b_i$, whereas the dual cost satisfies 
\[
\begin{array}{ll}
\boldsymbol{\Delta D} &\leq (b_i/B_i - z_i^c)\cdot B_i + \sum_{i'\in N(j)\setminus\{i\}}\left(\frac{1}{d-1}\right) \cdot (z_{i'}^c+C\cdot b_{i'}/B_{i'})\cdot B_{i'} \\

& \leq b_i - z_i^c\cdot B_i + (d-1)\cdot \left(\frac{1}{d-1}\right)\cdot(z_{i} ^c \cdot B_{i}+C\cdot b_{i}) \boldsymbol{\leq (1+C)\cdot b_{i}}
\end{array}
\]
As in Theorem \ref{thm:vw-high-deg}'s proof, $z_i^c=C\cdot \frac{b_i}{B_i}\big(\big(\frac{d}{d-1}\big)^{d^c(i)}-1\big)$, where $d^c(i)$ is the degree of $i$'s current copy, or equivalently, the number of $i$'s edges since $z_i^c$ was last nullified. Hence, by our choice of $C$, after at most $k$ $i$-edges, $z_i^c=\frac{b_i}{B_i}$ (whether or not $i$ is matched), and $z_i$ is increased by $\frac{b_i}{B_i}$. As $d(i)\geq k\cdot \frac{B_i}{b_i}$ by the end, $z_i\geq 1$ for all $i$. The theorem follows.
\end{proof}

\subsection{General Bids}\label{subsec:gen-bids}
A natural way to extend Algorithm \ref{alg:primal-dual-ad-allocation} to general bids would be to replace for every ad slot $j$ and every neighbor $i$ (or $i'$) all appearances of $b_i$ (or $b_{i'}$) by $b_{ij}$ (resp., $b_{i'j}$) in the choice of $j$'s match and updates to $z_i^c$, $z_{i'}^c$ and $z_i$. Such dual updates would guarantee, similarly to our prior algorithms, that an advertiser $i$ with budget $B_i$ and rejected bids $b_{i0},b_{i1},\dots,b_{it}$ since its last match (ordered chronologically) would have dual variable 
\begin{equation}\label{eqn:dual}
z_i^c=\frac{1}{d-1}\cdot C \cdot \sum_{r=0}^{t} \frac{b_{i r}}{B_i}\cdot \left(\frac{d}{d-1}\right)^{t-r}
\end{equation}
Unfortunately, replacing $b_i$ by $b_{ij}$ in the updates for matched $i$ could result in $z_i$ arbitrarily small. Worse still, since previously-rejected bids may be greater than the current bid, setting $z_i^c$ to $\frac{b_{ij}}{B_i}$ could even \emph{decrease} $z_i^c$, complicating the task of bounding the primal-dual ratio. Algorithm \ref{alg:primal-dual-ad-allocation-general} below sidesteps these issues by considering bounded fractions of $z_i^c$, and using the following notation, motivated by Equation \ref{eqn:dual}, to represent variables $z_i^c$, and $z_i^f$ (the $f$ in the notation refers to a bounded fraction of $z_i^c$ ``used''). This notation's use will become apparent during the algorithm's analysis.

\begin{Def}
Let $z=\frac{1}{d-1}\cdot C \cdot \sum_{r=0}^t b_{r}\cdot \big(\frac{d}{d-1}\big)^r$. We think of $z$ as a number in base $\frac{d}{d-1}$, denoting it by $z=[b_t,\dots,b_1,b_0]$, disregarding the $\frac{1}{d-1}\cdot C$ term for simplicity. Addition and subtraction of numbers in this notation is done place-wise, disallowing carries/borrows. In particular, if $z=[b_t,\dots,b_1,b_0]$, then $z\cdot \frac{d}{d-1}+\frac{1}{d-1}\cdot C\cdot b=[b_t,\dots, b_1,b_0,b]$. Comparisons involving numbers in this notation refer to their numerical value.
\end{Def}

\begin{algorithm}[h]
	\caption{Online Ad Allocation in $(k,d)$-graphs with general bids.}
	\label{alg:primal-dual-ad-allocation-general}
	\begin{algorithmic}[1]
		\medskip
		\STATE \textbf{Init:} set $z_i\leftarrow 0\ , z_i^c\leftarrow 0\ \text{ for all } i\in L$
		\FORALL{$j\in R$}
		\IF{$j$ has a feasible neighbor $i$}
		
			\FORALL{feasible neighbors $i$}
				\STATE let $z_i^c=[b_{k-1},\dots,b_1,b_0]$.
				\STATE set $z_i^f\leftarrow [\min\{b_{k-1},b_{ij}/B_i\},\dots, \min\{b_{1},b_{ij}/B_i\}, \min\{b_{0},b_{ij}/B_i\}]$.
				\STATE set $z_i^c\leftarrow z_i^c-z_i^f$.
			\ENDFOR
			
			\STATE match $j$ to feasible neighbor $i$ maximizing $z_i^f\cdot B_i+C\cdot b_{ij}$.
			\STATE set $x_{i,j}\leftarrow 1$. 
			\STATE set $z_i^f\leftarrow 0$.
			\STATE set $z_i\leftarrow z_i+b_{ij}/B_i$.\label{line:general-match}
			
			\FORALL{feasible neighbor of $j$, $i'\neq i$}
				\STATE set $z_{i'}^f\leftarrow z_{i'}^f\cdot\big(\frac{d}{d-1}\big) +\frac{1}{d-1}\cdot C\cdot b_{i'j}/B_{i'}$
				\STATE $z_{i'}^c\leftarrow z_{i'}^c + z_{i'}^f$.
				\IF{$z_{i'}^c=[b_k,b_{k-1},\dots ,b_1,b_0]$ with $b_k\neq 0$}\label{line:k-digit-Loop}
					\STATE set $z_{i'}\leftarrow z_{i'} + b_k \cdot \frac{1}{k}$.
					\STATE set $z_{i'}^c\leftarrow [b_{k-1},\dots,b_1,b_0]$.\label{line:k-digit-LoopEnd}
				\ENDIF
				\IF{$z_{i'}^c=[b_{k-1},\dots,b_1,b_0]$ with all digits $b_r\neq 0$}\label{line:k-1-digit-loop}
					\STATE let $b=\min\{b_r\}_{r=0}^{k-1}$.
					\STATE set $z_{i'}\leftarrow z_{i'} + b$.
					\STATE set $z_{i'}^c\leftarrow [b_{k-1}-b,\dots,b_1-b,b_0-b]$.\label{line:k-1-digit-loopEnd}
				\ENDIF
			\ENDFOR
		\ENDIF
		\ENDFOR
		
		\FORALL{$i\in L$}\label{line:feasibilityLoop}
		\STATE set $z_i\leftarrow \max\{1,z_i\}$.\label{line:feasibility}
		\ENDFOR
	\end{algorithmic}
\end{algorithm}

The algorithm for the general bids setting is Algorithm \ref{alg:primal-dual-ad-allocation-general}, below. The algorithm's primal feasibility is trivial, as is its dual feasibility, due to Lines \ref{line:feasibilityLoop}-\ref{line:feasibility}. It remains to bound the ratio of the cost of the dual solution to the value of the primal solution. 

\paragraph{High-Level intuition:}\label{par:intuition} The algorithm asserts three invariants.~The first guarantees increases in $z_i$ are ``paid for'' by increases in $z_i^c$, allowing us to focus on bounding changes to $z_i^c$.~A second invariant guarantees every increase of $z_i$ by some value $b/B_i$ can be accredited to previous bids (or fractions thereof) of total value at most $k\cdot b/B_i$. As the graph is $(k,d)$, if every bid of $i$ of value $b$ were to cause $z_i$ to increase (by at least $b/(k\cdot B_i)$, by the above), then eventually $z_i\geq1$. However, some bids may not incur an increase in $z_i$. The third and last invariant guarantees the total value of fractions of bids that do not cause $z_i$ to increase is at most $k\cdot R_{\max}$, and so $z_i\geq (1-R_{\max})$ before Lines \ref{line:feasibilityLoop}-\ref{line:feasibility}. Thus, the cost of rounding each $z_i$ to one in these lines is at most $R_{\max}/(1-R_{\max})$ of the previously-paid dual cost. The bound will follow. The following four lemmas formalize the above, allowing us to derive our sought-after bound.

\begin{lem}\label{lem:digital}
Before every ad slot's arrival and before Line \ref{line:feasibilityLoop}, every $z_i^c$ is a number in the above numeral system satisfying the following three properties:
\begin{enumerate}[(i)]
\itemsep0em 
\item $z_i^c$ is a $k$-digit number; i.e., $z_i^c=[b_{k-1},\dots,b_1,b_0]$.\label{item:k-digit}
\item $z_i^c$ has at most $k-1$ non-null digits.\label{item:k-1non-null}
\item Each digit of $z_i^c$ is no greater than $\max_j\{\frac{b_{ij}}{B_i}\}$.\label{item:small-digits}
\end{enumerate}
\end{lem}
\begin{proof}
Properties \ref{item:k-digit} and \ref{item:k-1non-null} are enforced explicitly by Lines \ref{line:k-digit-Loop}-\ref{line:k-digit-LoopEnd} and \ref{line:k-1-digit-loop}-\ref{line:k-1-digit-loopEnd}, respectively. Property \ref{item:small-digits} follows by induction: When $z_i^f$ is subtracted from $z_i^c$, every digit of $z_i^c$ is either nullified, if it was smaller than $b_{ij}/B_i$, or decreased by $b_{ij}/B_i$. After $z_i^f$ is updated and added to $z_i^c$, each digit of $z_i^c$ is increased by at most $b_{ij}/B_i$. Thus each digit is no greater than its previous value and $b_{ij}/B_i$, both of which are at most $\max_j\{\frac{b_{ij}}{B_i}\}$.
\end{proof}

\begin{lem}\label{lem:overflow}
If $k\geq d-1$ and $C=1/\big(\big(\frac{d}{d-1}\big)^k-1\big)$, every increase in $z_i$ by some $b$ in Lines \ref{line:k-digit-Loop}-\ref{line:k-digit-LoopEnd} and \ref{line:k-1-digit-loop}-\ref{line:k-1-digit-loopEnd} goes hand-in-hand with both
\begin{enumerate}[(i)]
\itemsep0em 
	\item a decrease of the same value or higher in $z_i^c$, and \label{item:dec-zic}
	\item a decrease of $k$ times this value or less in the sum of digits of $z_i^c$.\label{item:dec-digits-zic}
\end{enumerate}
\end{lem}
\begin{proof}
In Lines \ref{line:k-digit-Loop}-\ref{line:k-digit-LoopEnd}, $z_i$ is increased by $b_k/k$. On the other hand, we remove $b_k$, the $k$-th digit of $z_i^c$ in this numeral system, resulting in a decrease of $z_i^c$ by
\[
\frac{1}{d-1}\cdot C \cdot b_k \cdot \left(\frac{d}{d-1}\right)^k\geq \frac{1}{k}\cdot b_k.
\]
Thus, Properties \ref{item:dec-zic} and \ref{item:dec-digits-zic} both hold for Lines \ref{line:k-digit-Loop}-\ref{line:k-digit-LoopEnd}.
In Lines \ref{line:k-1-digit-loop}-\ref{line:k-1-digit-loopEnd}, the value of $z_i^c$ is decreased by $\frac{1}{d-1}\cdot C\cdot \sum_{r=0}^{k-1} b \cdot \big(\frac{d}{d-1}\big)^r=C\cdot \big(\big(\frac{d}{d-1}\big)^k-1\big) \cdot b$, which is exactly $b$, by our choice of $C$. The decrease in the sum of digits of $z_i^c$ on the other hand is exactly $k\cdot b$.
\end{proof}

\begin{lem}\label{lem:match}
Taking $C=1/\big(\big(\frac{d}{d-1}\big)^k-1\big)$ guarantees every increase in $z_i$ by $b_{ij}/B_i$ in Line \ref{line:general-match} coincides with a decrease of at most $b_{ij}/B_i$ in $z_i^c$. Moreover, $\Delta digit$, the decrease in sum of digits of $z_i^c$, satisfies $\Delta digit + b_{ij}/B_i \leq k\cdot b_{ij}/B_i$.
\end{lem}
\begin{proof}
In Line \ref{line:general-match}, $z_i^f$, which was subtracted from $z_i^c$, is nullified. Both bounds follow similarly to our proof of Lemma \ref{lem:overflow} relying on $z_i^f$ being a $k$-digit number with at most $k-1$ non-null digits, by Lemma \ref{lem:digital}, and each digit of $z_i^f$ being no greater than $b_{ij}/B_i$, by initialization of $z_i^f$.
\end{proof}

\begin{lem}\label{lem:almost-feasible}
By Line \ref{line:feasibilityLoop} each $i$ satisfies $z_i\geq \frac{\sum_j b_{ij}-k\cdot \max_j\{b_{ij}\}}{k\cdot B_i}\geq 1-\frac{\max_{j}b_{ij}}{B_i}\geq  1-R_{\max}$.
\end{lem}
\begin{proof}
Throughout the algorithm, every edge $(i,j)$ causes the sum of digits of $z_i^c$ to increase by $b_{ij}/B_i$ (again ignoring the $\frac{1}{d-1}\cdot C$ term), unless $(i,j)$ are matched. Moreover, the sum of digits does not decrease due to carries. On the other hand, every increase in $z_i$ by $b$ coincides with a decrease in the sum of digits of $z_i^c$ plus $\sum_{(i,j)\textnormal{ matched}}b_{ij}/B_i$, of at most $k\cdot b$, by Lemmas \ref{lem:overflow} and \ref{lem:match}. 
Put otherwise, the increase in $z_i$ is at least $1/k$ times the total sum of $i$'s bids so far, minus the sum of digits of $z_i^c$. By Lemma \ref{lem:digital}, the sum of digits of $z_i^c$ by Line \ref{line:feasibilityLoop} cannot exceed $k\cdot \max_j\{b_{ij}/B_i\}$. The lemma follows.
\end{proof}

Given the above we can now prove our main result.

\begin{thm}\label{thm:adwords_lb}
On general-bid ad allocations on $(k,d)$-graphs with $k\geq d-1$ Algorithm \ref{alg:primal-dual-ad-allocation-general} gains $\sum_i \big(B_i - \max_{j}b_{ij} \big)\cdot\big(1-\big(1-\frac{1}{d}\big)^k\big)$, and is thus $\big(1-R_{\max}\big)\cdot\big(1-\big(1-\frac{1}{d}\big)^k\big)$-competitive.
\end{thm}
\begin{proof}
Lemmas \ref{lem:overflow} and \ref{lem:match} imply increases in $z_i$ can be traced back to a previous increase in $z_i^c$ of the same value or higher. We therefore bound increases of $z_i^c\cdot B_i$ in order to bound the total dual cost.
For each online $j\in R$, by our choice of match $i$, the change to the dual cost is at most $(1+C)$ times the change to the primal value, as in Algorithm \ref{alg:primal-dual-ad-allocation}. However, by Lemma \ref{lem:almost-feasible}, by Line \ref{line:feasibilityLoop} each $i$ satisfies $z_i\geq (1-\max_{j}b_{ij}/B_i)$. Consequently, we have that before Line \ref{line:feasibilityLoop} the primal value $P$ and dual cost $D$ satisfy
\[
P\geq \frac{1}{1+C}\cdot D\geq \sum_i \big(B_i - \max_{j\in N(i)}b_{ij} \big)\cdot \left(1-\left(1-\frac{1}{d}\right)^k\right)
\]
As the primal value is unaffected by Lines \ref{line:feasibilityLoop}-\ref{line:feasibility}, $P$ above is our algorithm's gain. The competitive ratio follows from $OPT\leq \sum_i B_i$ and the definition of $R_{\max}$.
\end{proof}

Finally, we note that Lemmas \ref{lem:digital},\ref{lem:overflow},\ref{lem:match} and \ref{lem:almost-feasible} hold for all advertisers $i$ satisfying $\sum_j b_{ij} \geq k\cdot B_i$, irrespective of outliers who don't hold this property, implying the following.

\begin{thm}[Outliers]\label{thm:outliers_general_ad_allocation}
	Let $S\subseteq L$ be the set of outlying advertisers (advertisers $i$ with $\sum_j b_{ij} < k\cdot B_i$), and $\alpha$ be such that $\sum_{i\in S} B_i \leq \alpha \cdot \sum_{i\in L}$. Then Algorithm \ref{alg:vw-high-pd} gains at least 
	\[
	(1-\alpha)\cdot\left(\big(1-R_{\max}\big)\cdot (1-\left(1-\frac{1}{d}\right)^k)\right)\cdot \big(\sum_i B_i\big),
	\]
	and in particular is $(1-\alpha)\cdot\left(\big(1-R_{\max}\big)\cdot (1-\left(1-\frac{1}{d}\right)^k)\right)$-competitive.	
\end{thm}

\section{Upper Bounds}\label{sec:upper-bounds}
\subsection{Maximum Matching and Vertex-Weighted Matching}

In order to construct hard examples, we start by showing that the optimal matching in $(k,d)-$graphs matches \emph{all} the advertisers whenever $k\geq d$.

\begin{lem}\label{lem:saturate}
	Every $(k,d)-$graph $G=(L,R,E)$ with $k\geq d$ has a matching matching all of $L$. 
\end{lem}
\begin{proof}
	By Hall's Theorem $G$ has a matching with all of $L$ matched if and only if every subset $A\subseteq L$ satisfies $|\Gamma(A)|\geq |A|$. But, as $G$ is a $(k,d)-$graph we have
	\begin{equation*}
	k \cdot |A| \leq |E[G[A]]| \leq d\cdot |\Gamma(A)|
	\end{equation*}
	Consequently, we find that $|\Gamma(A)|\geq \frac{k}{d} \cdot|A| \geq |A|$, and the lemma follows.
\end{proof}

Equipped with Lemma \ref{lem:saturate} we may now prove this section's main result -- an upper bound matching the lower bounds of Section \ref{sec:high_deg}, implying algorithm \textsc{high-degree}'s optimality. To this end we cause \textsc{high-degree} to be effectively indistinguishable from any other algorithm.

\begin{thm}\label{thm:matching_ub}
	For all $k\geq d$ no deterministic online algorithm for bipartite matching can achieve competitive ratio better than $1-(1-\frac{1}{d})^k$ on $(k,d)-$graphs.
\end{thm}

\begin{proof}
	Let $\mathcal{A}$ be some online matching algorithm. The adversarial input consists of $d^{k+1}$ advertisers, with the ad slots arriving in $k$ phases, numbered $0$ to $k-1$. During the $i$-th phase, which begins with $d^{k+1}\cdot (1-\frac{1}{d})^i$ unmatched advertisers each of degree $i$, the arriving ad slots each have exactly $d$ neighbors, all unmatched; every unmatched advertiser neighbors exactly one new ad slot per phase.~Every phase causes unmatched advertisers to have their degree increase by one, and exactly a $(1-\frac{1}{d})$-fraction of the advertisers unmatched at the phase's beginning remain unmatched. (If algorithm $\mathcal{A}$ does not match some ad slot to one of its $d$ unmatched neighbors, we consider it matched to an arbitrary neighbor; this can only serve to improve $\mathcal{A}$'s performance.) After the $k$ phases additional ad slots of degree exactly $d$ arrive in order to increase the degree of the matched advertisers to $k$. The resulting graph is $k$-regular and $d$-regular on the offline and online sides respectively, and is thus a $(k,d)-$graph. Moreover, exactly $d^{k+1}\cdot (1-\frac{1}{d})^k$ of the $d^{k+1}$ advertisers are unmatched. However, by Lemma \ref{lem:saturate} all $d^{k+1}$ advertisers can be matched simultaneously. The theorem follows.
\end{proof}


\begin{cor}\label{cor:regular_hard}
	The bound of Theorem \ref{thm:matching_ub} holds for $d$-regular graphs with $d^{d+1}\leq n$, where $n=|L|=|R|$. In particular, Algorithm \textsc{high-degree} is optimal for $d$-regular graphs with $d=O\left(\frac{\log n}{\log \log n}\right)$.
\end{cor}

\subsection{Upper Bound for Ad Allocation}\label{sec:1-R}
In this subsection we prove an upper bound on the possible competitive ratio of deterministic algorithms in $(k,d)$-graphs. We start by showing a simple weaker bound, useful in proving this section's main result.

\begin{lem}\label{lem:1-R}
For all ratio $R_{\max}$ no deterministic algorithm can achieve competitive ratio better than $(1-R_{\max})$ for the ad allocation problem under the adversarial model. This bound holds even for $(k,d)$-graphs for all $k$ and $d$.
\end{lem}

\begin{proof}
The hard input consists of disjoint stars with advertisers for internal vertices and ad slots for leaves. Every advertiser $i$ has budget $B_i=1$, with $i$'s bids given by
\[
b_{ij} =
\begin{cases}
R_{\max} & \text{if } i\text{'s remaining budget is less than } R_{\max} \\
\epsilon  & \text{else} \\
\end{cases}
\]
Given enough ad slots, an optimal allocation exhausts all advertisers' budgets, but every advertiser $i$ gains at most $1-R_{\max}+\epsilon$, whether or not $i$ has neighbors $j$ with $b_{ij}=R_{\max}$. Summing over all advertisers, the lemma follows.
\end{proof}
Using the above and extending Theorem \ref{thm:matching_ub}'s proof, we can now prove the following.

\begin{thm}\label{thm:adwords_ub}
	For all $k\geq d$ no deterministic online algorithm for ad allocation is better than $(1-R_{\max})\cdot \left(1-\left(1-\frac{1}{d}\right)^{k/R_{\max}}\right)$-competitive on $(k,d)$-graphs with $R_{\max}\leq \frac{1}{2}$ a unit fraction.
\end{thm}

\begin{proof}
The offline side consists of $d^{k/R_{\max}}$ advertisers, each with a budget $B_i=1$. For the first phase, all edges have bids $R_{\max}$. During $k/R_{\max}$ rounds ad slots arrive, each neighboring $d$ distinct advertisers, and a $\frac{1}{d}$-fraction of the advertisers are matched. The next round is as the last, but restricted to the previously unmatched advertisers. There are $(1-\frac{1}{d})^{k/R_{\max}}$ unmatched advertisers by this phase's termination; these advertisers now satisfy the offline side's constraints for $(k,d)$-graphs, and receive no more neighbors. All of these advertisers' potential profit is lost. For the matched advertisers we now apply the construction of Lemma \ref{lem:1-R} to guarantee that at most a $(1-R_{\max})$-fraction of their potential profit in an optimal solution is gained, for a total gain of $(1-R_{\max})\cdot \left(1-\left(1-\frac{1}{d}\right)^{k/R_{\max}}\right)\cdot |L|$. Applying Lemma \ref{lem:saturate} repeatedly we find that there exists an allocation with all advertisers  unmatched by algorithm $\mathcal{A}$ matched $1/R_{\max}$ times (to neighbors for which they bid $R_{\max}$), and all advertisers matched by $\mathcal{A}$ also exhausting their budgets simultaneously. The theorem follows.
\end{proof}

\newpage
\section{Randomized Algorithms}\label{sec:randomization}
By relying on the dual updates of Algorithm \ref{alg:vw-high-pd}, we prove competitiveness of algorithm \textsc{random}, which matches every arriving ad slot to some feasible neighbor.
\begin{thm}\label{thm:rand-matching}
	Algorithm \textsc{random} achieves expected competitive ratio of $1-\left(1-\frac{1}{d}\right)^k$ for both unweighted and vertex-weighted matching problems.
\end{thm}

\begin{proof} 
	Consider some maximally-matching algorithm $\mathcal{A}$.\footnote{By ``maximally-matching'' we mean algorithm $\mathcal{A}$ always matches an arriving ad slot, if at all possible.} We maintain and update a dual solution as in our deterministic Algorithm \ref{alg:vw-high-pd} while choosing matches according to Algorithm $\mathcal{A}$. As observed in the proof of Theorem \ref{thm:vw-high-deg}, such dual updates guarantee all unmatched advertisers $i\in L$ with current degree $d(i)$ satisfy
	\[
	z_i=C\cdot\left(\left(\frac{d}{d-1}\right)^{d(i)}-1\right)
	\]
	Consequently, these dual update rules guarantee dual feasibility for any maximally-matching algorithm, including \textsc{random}, provided $C= 1/((\frac{d}{d-1})^k-1)$.
	We need only bound the expected ratio between the dual and primal solutions' values.
	
	Consider some ad slot $j$ matched to some $i$. We recall that in the vertex-weighted matching problem the bid $b_{ij}$ is exactly $b_{ij}=B_i$. Therefore, given the current state (determined by the previous random choices), including the set $N_F(j)$ of $j$'s unmatched (feasible) neighbors, $j$'s match is chosen uniformly among $N_F(j)$ by \textsc{random}, and consequently
	\[
	\mathbb{E}[\Delta P | state]=\frac{1}{|N_F(j)|}\cdot \sum_{i\in N_F(j)}B_i.
	\]
	On the other hand, by the same argument
	\[
	\begin{array}{ll}	
	\mathbb{E}[\Delta D | state] & =
	\frac{1}{|N_F(j)|}\cdot \sum\limits_{i\in N_F(j)}\left((1-z_i)\cdot B_i + 
	\sum\limits_{i'\in N_F(j)\setminus\{i\}} \left(\frac{1}{d-1}\cdot (z_{i'}+C)\cdot B_{i'}\right)\right) \\
	
	& =
	\frac{1}{|N_F(j)|}\cdot \sum\limits_{i\in N_F(j)}\left((1-z_i)\cdot B_i\right) 
	+ \frac{1}{|N_F(j)|}\cdot \sum\limits_{i'\in N_F(j)} \left(\frac{|N_F(j)|-1}{d-1}\cdot (z_{i'}+C)\cdot B_{i'}\right) \\
	
	& \leq \frac{1}{|N_F(j)|}\cdot \sum\limits_{i\in N_F(j)} B_i\cdot (1+C). \\
	\end{array}
	\]
	With the last inequality following from $|N_F(j)|\leq |N(j)|\leq d$. Taking total expectation over the possible states, we obtain $\mathbb{E}[\Delta D]\leq (1+C)\cdot \mathbb{E}[\Delta P]$. The theorem follows.
\end{proof}

We note that Theorem \ref{thm:rand-matching} can also be proved using the potential-based proof of Subsection \ref{sec:potential}, observing that the expected potential change incurred by the processing of every online arrival is non-negative. In addition, in the same way that Theorem \ref{thm:vw-high-deg} is extended in Theorem \ref{thm:outliers_vw}, we can show that \textsc{random} is also robust to outliers. We omit the details for brevity. Finally, we show that \textsc{random} also performs well for the general online ad allocation problem.

\begin{thm}
	Algorithm \textsc{random} achieves expected competitive ratios of  $1-\left(1-\frac{1}{d}\right)^k$ and $(1-R_{\max})\cdot\big(1-\big(1-\frac{1}{d}\big)^k\big)$ for the equal-bids and general-bids ad allocation problems.
\end{thm}

\begin{proof}[Proof (sketch)]
	The proof resembles that of Theorem \ref{thm:rand-matching}, relying on Algorithms \ref{alg:primal-dual-ad-allocation} and \ref{alg:primal-dual-ad-allocation-general} respectively for the dual-fitting analysis. Dual feasibility is guaranteed by the dual updates. On the other hand, linearity of expectation implies the expected primal-dual ratio matches that of Algorithms \ref{alg:primal-dual-ad-allocation} and \ref{alg:primal-dual-ad-allocation-general} (for the latter, this requires showing Lemmas \ref{lem:digital}--\ref{lem:almost-feasible} all hold in expectation). The claimed bounds follow.
\end{proof}

\section{Future Work and Open Questions}\label{sec:discuss}
This paper attempts to give a theoretical explanation of the empirical success of simple heuristic algorithms for online ad allocation in practice and further proposes better algorithms under assumptions that could explain above-mentioned success. The paper further poses several interesting follow-up questions.

\textbf{Optimality for Adwords.} We proved optimality of our algorithms for online maximum and vertex-weighted matching. However, for the general ad allocation problem our lower and upper bounds differ by a factor of  $\big(1-\big(1-\frac{1}{d}\big)^k\big)/\big(1-\big(1-\frac{1}{d}\big)^{k/R_{\max}}\big)$. For small $R_{\max}$ (i.e., the AdWords problem), this discrepancy is large. Can better algorithms be obtained for this problem, or can the upper bounds be tightened (or both)?

\textbf{Randomization.} We are able to prove that a multitude of randomized algorithms, which we do not discuss for the sake of brevity, match our deterministic algorithms' expected competitive on $(k,d)$-graphs. However, we have no upper bounds on randomized algorithms' performance. We believe that randomization does allow for better results, and have some preliminary results that indicate this is indeed the case.

\textbf{Stochastic Models.} An interesting direction would be to extend our analysis of $(k,d)$-graphs to stochastic models, in which it seems plausible that even better competitiveness guarantees should be achievable.

\section{Acknowledgments}
This work was supported in part by United States-Israel BSF Grant No. 2010-246 and ISF Grant No. 954/11.

\bibliographystyle{acmsmall}
\bibliography{online-matching}


\newpage
\appendix

\section{Bad Instance for Algorithm \textsc{greedy} for General Ad Allocation}\label{sec:greedy-hard}
We now prove our analysis of Algorithm \textsc{greedy} for the ad allocation problem is tight.

\begin{thm}
For all $k\geq d-1$ and $R_{\max}\leq \frac{1}{2}$ a unit fraction, there exist $(k,d)$ ad allocation instances for which algorithm \textsc{greedy} can achieve competitive ratio exactly $\frac{(1-R_{\max})\cdot k}{k+(d-1)\cdot(1-R_{\max})}$.
\end{thm}

\begin{proof}
Let $1-R_{\max}=\frac{a}{b}$ for $0<a<b$ and $a$ and $b$ integers. The hard instance will consist of $k\cdot b + (d-1)\cdot a=b\cdot(k+(d-1)\cdot(1-R_{\max}))$ advertisers. Each advertiser has budget exactly one. We designate $k\cdot b$ advertisers to be the ``lucky'' advertisers, from which we will achieve revenue of $(1-R_{\max})$ and the remaining $(d-1)\cdot a$ ``unlucky'' advertisers will garner no profit. The theorem will follow by constructing the instance such that all budgets can be exhausted simultaneously.

All edges have bids either $R_{\max}$ or some arbitrarily small positive $\epsilon$. At first, each arriving ad slot will have $d$ edges with bids $R_{\max}$, one to some lucky advertiser of lowest degree (to whom the ad slot is matched), and $(d-1)$ edges to some unlucky advertisers of lowest degree. After $\frac{a\cdot k\cdot b}{b-a}$ ad slots arrive (this value is integral, as is $\frac{1}{R_{\max}}=\frac{b}{b-a}$), the following holds
\begin{enumerate}[(i)]
\item each of the unlucky advertisers are unmatched and have degree exactly $k\cdot \frac{b}{b-a}=k\cdot \frac{1}{R_{\max}}$.
\item each of the lucky ad slots are matched to all of their neighbors, and have degree exactly $\frac{a}{b-a}=\frac{1}{R_{\max}}-1$.
\end{enumerate}
The remaining ad slots recreate the construction of Lemma \ref{lem:1-R}, thus guaranteeing each of the lucky advertisers gain no more than $1-R_{\max}+\epsilon$. On the other hand all the lucky advertisers can exhaust their budgets without using any of the $R_{\max}$-bid edges of ad slots neighboring unlucky advertisers, which, as can be readily verified (using, e.g. Lemma \ref{lem:saturate} repeatedly), allows both lucky and unlucky advertisers to exhaust their budgets simultaneously whenever $k\geq d-1$. The described instance is $(k,d)$ and the theorem follows.
\end{proof}

The above bound holds for \emph{any} $R_{\max}\leq \frac{1}{2}$, as the following theorem asserts.

\begin{thm}
	For all $k\geq d-1$ and $R_{\max}\leq \frac{1}{2}$ there exist $(k,d)$ ad allocation instances for which algorithm \textsc{greedy} can achieve competitive ratio exactly $\frac{(1-R_{\max})\cdot k}{k+(d-1)\cdot(1-R_{\max})}$.
\end{thm}

\begin{proof}[Proof (sketch)]
In order to generalize the above, we rely on the fact that every number $R_{\max}$ in the range $(0,\frac{1}{2}]$ can be written as a convex combination of two unit fractions $\frac{1}{a}\leq R_{\max}\leq \frac{1}{b}$. That is, $w_a \cdot \frac{1}{a} + w_b\cdot \frac{1}{b} = R_{\max}$ and $w_a+w_b=1$. We glue $2n$ copies of the above construction at the advertisers, $n$ of the copies with budget $w_a/n$ ($w_b/n$) for each advertiser, and highest bid-to-budget ratio in the copy being $1/a$ (resp. $1/b$). In this case the overall budget from all copies is $n\cdot (w_a/n+w_b/n)=w_a+w_b=1$, and for large enough $n$ each bid is at most $w_b/(b\cdot n)<R_{\max}$. On the other hand, all unlucky advertisers are completely unmatched and garner no profit, and all lucky advertisers gain a total of $n\cdot (w_a/n+w_b/n - w_a/(a\cdot n) + w_b/(b\cdot n)) = 1-(w_a/a+w_b/b)=1-R_{\max}$. As in Lemma \ref{lem:1-R} we can guarantee each such lucky advertiser yields at most $\epsilon$ additional revenue.
\end{proof}


\end{document}